\documentclass[sigconf,acmthm=true,authorversion]{acmart}

\usepackage{booktabs} % For formal tables
\usepackage{caption,subcaption}

\acmYear{2018}
\copyrightyear{2018}
\setcopyright{acmlicensed}
\acmDOI{10.1145/3178126.3178142}
\acmISBN{978-1-4503-5642-8/18/04}

\acmBooktitle{HSCC '18: 21st International Conference on
Hybrid Systems: Computation and Control (part of CPS Week),\phantom{0.2cm} April
11--13, 2018, \phantom{0.2cm} Porto, Portugal}

\acmConference[HSCC '18]{HSCC '18: 21st International Conference on
Hybrid Systems: Computation and Control (part of CPS Week)}{April
11--13, 2018}{Porto, Portugal}

\def\r{\mathbb{R}}

\def\be{\begin{equation}}
\def\ee{\end{equation}}
\def\ben{\begin{equation*}}
\def\een{\end{equation*}}
\newcommand{\dfb}{\stackrel{\Delta}{=}}

\def\u{\mathcal{U}}

\begin{document}
\title{Lyapunov Design for Event-Triggered Exponential Stabilization}
\titlenote{The work has been supported by NWO Domain TTW, The Netherlands, under the project TTW\#13712
``From Individual Automated Vehicles to Cooperative Traffic Management -- predicting the benefits of automated driving
through on-road human behavior assessment and traffic flow models'' (IAVTRM).
}

\author{Anton V. Proskurnikov}
\affiliation{%
  \institution{Delft University of Technology}
  \city{Delft}
  \state{Netherlands}
}
\email{anton.p.1982@ieee.org}

\author{Manuel Mazo Jr.}
\affiliation{%
  \institution{Delft University of Technology}
  \city{Delft}
  \state{Netherlands}
}
\email{m.mazo@tudelft.nl}

\begin{abstract}
Control Lyapunov Functions (CLF) method gives a constructive tool for stabilization of nonlinear systems.
To find a CLF, many methods have been proposed in the literature, e.g. backstepping for cascaded systems and sum of squares (SOS) programming for polynomial systems. Dealing with continuous-time systems, the CLF-based controller is also continuous-time, whereas practical implementation on a digital platform requires sampled-time control. In this paper, we show that if the continuous-time controller provides exponential stabilization, then an exponentially stabilizing \emph{event-triggered} control strategy exists with the convergence rate arbitrarily close to the rate of the continuous-time system.
\end{abstract}


\maketitle

\section{Introduction}

The idea to use Lyapunov functions as a control \emph{design} tool~\cite{KalmanBertram:1960} naturally leads to
the method of Control Lyapunov Functions (CLF). Being a natural extension of the usual Lyapunov functions
for controlled systems, a CLF is a function that becomes a Lyapunov function of the closed-loop system under an appropriate choice of the controller. The existence of a CLF is necessary and sufficient for stabilization
of a general nonlinear system, as implied by the fundamental Artstein's theorem~\cite{Artstein:1983}. This theorem, however, provides no constructive way to design the stabilizing control, moreover, this control in general can be ``relaxed'' (randomized), mapping a system's state into a probability measure. These limitations may be overcome in the case of \emph{affine} systems. Sontag's theorem~\cite{Sontag:1989} gives an explicit formula for one stabilizing feedback, which appears to be continuous everywhere except for the equilibrium point.

Whereas to find a CLF for a given control system is a non-trivial problem, in some important situations it can be found in an explicit form. Examples include
some homogeneous systems~\cite{FauborgPomet:2000}, passive and feedback-passive systems~\cite{KokotovicArcak:2001,Khalil} and cascaded systems, for which both CLF and stabilizing controller can be
delivered by the \emph{backstepping} and \emph{forwarding} procedures~\cite{KrsticKokotovicBook,SepulchreKokotovicBook}.
CLFs can often be computed by using numerical tools, e.g. the Sum of Square (SOS) programming~\cite{Furqon:2017} and Zubov's method~\cite{Camilli:2008}.

Nowadays the method of CLF is recognized as a powerful tool in nonlinear control systems' design~\cite{KrsticKokotovicBook,Khalil,SepulchreKokotovicBook}.
A CLF gives a solution to the Hamilton-Jacobi-Bellman equation for an appropriate performance index, giving a solution to the \emph{inverse optimality} problem~\cite{FreemanKokotovic:1996}.
The method of CLF has been extended to discrete-time~\cite{KellettTeel:2004}, time-delay~\cite{Jancovic:2001}
and hybrid systems~\cite{Sanfelice:2013,Ames:2014}. Combining CLFs and Control Barrier Functions (CBFs), correct-by-design controllers for stabilization under safety constraints can be obtained~\cite{RomdlonyJayawardhana:2016,AmesTabuada:2017},
enabling to design \emph{safety-critical} control systems, arising e.g. in  automotive~\cite{AmesTabuada:2017,NilssonTabuada:2016} and aerospace~\cite{LuoChuLing:2005} applications.

Typically, CLF-based controllers are continuous-time. Their implementation on digital platforms requires to introduce time sampling. A straightforward approach, often used in engineering, is to emulate the continuous-time feedback by
a discrete-time control, sampled periodically at a high rate. However, rigorous techniques for nonlinear controllers' discretization have appeared only recently~\cite{NesicTeel:2004,ArcakNesic:2004} and are highly non-trivial. As an alternative to these techniques, digital controllers based on \emph{event-triggered} sampling can be used. Event-triggered sampling has a number of advantages over periodic (time-triggered) sampling, providing parsimonious use of communication and power~\cite{Astrom:2002,Tabuada:2007,BorgersHeemels:2014,Araujo:2014}.
%Typical approaches to the analysis of event-triggered control algorithms are based on general hybrid systems %theory~\cite{Goebel:2009}, switching systems~\cite{SelivanovFridman:2016} and
%delay systems~\cite{YueTianHan:2013,SelivanovFridman:2016-1}.

A natural question arises whether the existence of a (continuous-time) CLF enables one to design an event-triggered stabilizing controller. Such controllers have been found for only a few special cases. The most studied is the case where the system admits a so called \emph{ISS Lyapunov} function~\cite{Tabuada:2007}, being a special CLF that ensures a special \emph{input-to-state stability} (ISS)~\cite{SontagWang95} property of the closed-loop system.
A more recent result from~\cite{SeuretPrieurMarchand:2013} relaxes the ISS condition to a stronger version of usual asymptotic stability, however the control algorithm from~\cite{SeuretPrieurMarchand:2013}, in general, does not ensure positive dwell time between the consecutive events, nor even the absence of Zeno behaviors.
Another approach, based on Sontag's universal formula~\cite{Sontag:1989} and inheriting its basic limitations has been proposed in~\cite{Marchand:2013,Marchand:2013IFAC}.
All of these results rely on restrictive assumptions, discussed in detail in Section~\ref{sec.prelim},
and do not allow to estimate the convergence rate efficiently. In many situations a CLF can be designed that provides \emph{exponential} convergence rate~\cite{Ames:2014} in continuous time. A natural question arises whether \emph{event-based} controllers can provide the same (or an arbitrarily close) convergence rate. In this paper, we give an affirmative answer to this fundamental question. Under natural assumptions, we design an event-triggered controller, similar in structure to the one proposed in~\cite{SeuretPrieurMarchand:2013}, but
retaining the exponential convergence and providing a positive \emph{dwell time} between consecutive events.
%Furthermore, we design self-triggered and periodic event-triggered controllers, that do not need the constant monitoring of %the solution and enable flexible \emph{time scheduling} for the transmission of sensors' measurements and control %commands.

%\bibliographystyle{ACM-Reference-Format}
%\bibliography{sample-bibliography}

\section{Preliminaries and problem setup}\label{sec.prelim}

%Henceforth $\r^{m\times n}$ stands for the set of $m\times n$ real matrices, $\r^n=\r^{n\times 1}$.
Given a map $G:\r^n\to\r^m$ such that $G(x)=(G_1(x),\ldots,G_m(x))^{\top}\in\r^m$,
we use $G'(x)=\big(\frac{\partial G_i(x)}{\partial x_j}\big)$ to denote its $m\times n$ Jacobian matrix.

\subsection{Control Lyapunov functions in stabilization problems}

In this paper, we deal only with CLFs for global asymptotic stabilization of general nonlinear systems of the form
\be\label{eq.syst}
\dot x(t)=F(x(t),u(t)),\quad t\ge 0.
\ee
Here $x(t)\in\r^{n}$ stands for the state vector and $u(t)\in U\subseteq\r^{m}$ is the control input.
%(the case $U=\r^m$ corresponds to the absence of input constraints).
%We assume that the system~\eqref{eq.syst} has an equilibrium point at $x=0, u=u_0\in U$, that is
%\be\label{eq.equil}
%F(0,u_0)=0.
%\ee
Our goal is to find a controller $u(\cdot)=\mathcal U(x(\cdot))$, where $\mathcal U:x(\cdot)\mapsto u(\cdot)$ is some causal operator, such that for any $x(0)\in\r^{n}$ the closed-loop system has a forward complete (existing up to $+\infty$) solution, and all solutions converge to an equilibrium, assumed, without loss of generality, to be $0$
\be\label{eq.stab}
x(t)\xrightarrow[t\to\infty]{} 0\quad\forall x(0).
\ee

Following the definition from~\cite{Sontag:1989}, henceforth all CLFs are supposed to be radially unbounded, or \emph{proper}~\cite{Sontag:1989}.
\begin{definition}\cite{Sontag:1989}
A $C^1$-smooth function $V:\r^n\to\r$ is called a \emph{control Lyapunov function} (CLF) in the stabilization problem, if
\begin{gather}
V(0)=0,\quad V(x)>0\,\forall x\ne 0,\quad \lim_{|x|\to\infty}V(x)=\infty;\label{eq.pos-def}\\
\inf_{u\in U} V'(x)F(x,u)<0\quad\forall x\ne 0.\label{eq.inf-u}
\end{gather}
\end{definition}
\vskip 1mm
\par The condition~\eqref{eq.inf-u}, obviously, can be reformulated as follows
\be\label{eq.inf-u-2}
\forall x\ne 0\,\exists u=u(x)\in U\quad\text{such that $V'(x)F(x,u(x))<0$.}
\ee
If $F(x,u)$ is Lebesgue measurable (e.g., continuous), then the measurable selector theorem~\cite[Theorem~5.2]{Himmelberg:1975} implies that the function $u(x)$ can be chosen measurable. This function
is, in general, discontinuous, so that the closed-loop system has no
classical solutions. However, the existence of a CLF is necessary and sufficient~\cite{Artstein:1983} for the existence
of a \emph{relaxed} stabilizing control $x\mapsto v(x)$, where $v(x)$ is a probability distribution on $U$.

The situation becomes much simpler in the case where the system~\eqref{eq.syst} is \emph{affine}: $F(x,u)=f(x)+g(x)u$. Assuming that $f:\r^n\to\r^n$ and $g:\r^n\to\r^{n\times m}$ are continuous and $U$ is convex, the existence of a CLF ensures the possibility to design a controller $u=u(x)$, where $u:\r^n\to U$ is \emph{continuous} everywhere except for, possibly $x=0$ (the continuity at $x=0$ requires an additional assumption)~\cite{Artstein:1983}. Extending this result, Sontag~\cite{Sontag:1989} has proposed an explicit \emph{universal} formula, giving a broad class of stabilizing controllers. In the simplest case where $m=1$ (scalar control) and $U=\r$, Sontag's formula gives the following controller
\be\label{eq.sontag-scal}
u(x)=
\begin{cases}
-\frac{a(x)+\sqrt{a(x)^2+q(b(x))b(x)}}{b(x)},\,&b(x)\ge 0\\
0,&\text{otherwise.}
\end{cases}
\ee
Here the functions $a,b$ are defined as
\ben
a(x)\dfb V'(x)f(x),\quad b(x)\dfb V'(x)g(x)
\een
and $q(b)$ is a continuous function, $q(0)=0$. Controllers similar to~\eqref{eq.sontag-scal} have been proposed for $U$, being the Euclidean space $\r^m$ with $m>1$~\cite{Sontag:1989} and a closed ball in $\r^m$~\cite{Sontag:1991}.

\subsection{CLF and event-triggered control}

Dealing with continuous-time systems~\eqref{eq.syst}, Lyapunov controllers are also continuous-time,
which makes it impossible to implement them directly on a digital platform.
In reality, the control is always \emph{sampled time}, that is, the control command is computed and sent to the plant only at discrete instants $t_0=0<t_1<\ldots<t_n<\ldots$, remaining constant between them. The simplest time sampling is periodic $t_n=n\tau$. In spite of the belief that high-frequency periodic control (with small $\tau$) satisfactorily emulates the
continuous-time controller, mathematically rigorous analysis of the resulting nonlinear sampled-time system appears to be non-trivial~\cite{NesicTeel:2004,ArcakNesic:2004}. Alternatively, sampling can be triggered by some condition, or \emph{event}, e.g., $t_{n+1}$ can be the first instant after $t_n$ when the absolute value of the ``error'' $e(t)=x(t_n)-x(t)$ reaches a predefined threshold~\cite{Tabuada:2007}. This approach, known as event-based or event-triggered sampling has many advantages over periodic control, in particular,
it uses communication and energy resources parsimoniously~\cite{Astrom:2002,Tabuada:2007,BorgersHeemels:2014,Araujo:2014}.

A natural question arises whether a continuous-time CLF can be employed to design an \emph{event-triggered} stabilizing controller. Up to now, only a few results of this type have been reported in the literature. In the seminal work~\cite{Tabuada:2007}, an event-triggered controller requires the existence of a special CLF, called \emph{ISS Lyapunov function}, for which the conditions~\eqref{eq.pos-def},\eqref{eq.inf-u} are replaced by the following
\be\label{eq.iss}
\begin{gathered}
\alpha_1(|x|)\le V(x)\le \alpha_2(|x|)\quad\forall x\in\r^n\\
V'(x)F(x,k(x+e))\le -\alpha_3(|x|)+\gamma(|e|)\quad\forall x,e\in\r^n.
\end{gathered}
\ee
Here $\alpha_i(\cdot)$ ($i=1,2,3$) are $\mathcal K_{\infty}$-functions\footnote{That is, $\alpha_i$ are continuous, strictly increasing, $\alpha_i(0)=0$ and $\lim_{s\to\infty}\alpha_i(s)=\infty$.} and the mappings $u(\cdot):\r^n\to\r^m$, $F(\cdot,\cdot):\r^n\times\r^m\to\r^n$, $\alpha_3^{-1}(\cdot)$ and $\gamma(\cdot):\r_+\to\r_+$ are assumed to be locally Lipschitz. The continuous-time control $u=k(x)$ not only stabilizes the system, but in fact also provides \emph{input to state} stability (ISS) with respect to the measurement error $e$. The event-triggered controller, offered in~\cite{Tabuada:2007}, is as follows
\be\label{eq.tabuada}
\begin{gathered}
u(t)=k(x(t_n))\quad \text{if $t\in [t_n,t_{n+1})$}\\
t_0=0,\quad t_{n+1}=\inf\left\{t>t_n:\gamma(|e(t)|)=\sigma\alpha_3(|x(t)|)\right\},\\
e(t)=x(t_n)-x(t),\quad \sigma=const\in (0,1).
\end{gathered}
\ee
The controller~\eqref{eq.tabuada} provides positive \emph{dwell time} between consecutive events $\tau=\inf_{n\ge 0}(t_{n+1}-t_n)>0$, which is \emph{uniformly} positive for the solutions, starting in any compact set.

Whereas the condition~\eqref{eq.iss} holds for linear systems~\cite{Tabuada:2007} and some polynomial systems~\cite{AntaTabuada:2010}, in general it is restrictive and not easily verifiable.
Another approach to CLF-based design of event-triggered controllers have been proposed in~\cite{Marchand:2013,Marchand:2013IFAC}.
Discarding the ISS condition~\eqref{eq.iss}, this approach is based on Sontag's theory~\cite{Sontag:1989} and inherits its basic assumptions: first, the system has to be affine
$F(x,u)=f(x)+g(x)u$, where $f,g\in C^1$, second, the Sontag formula~\cite{Sontag:1989} gives an admissible controller, that is, $u(x)\in U$ for any $x$. The controllers from~\cite{Marchand:2013,Marchand:2013IFAC} also provide the positive dwell-time (or ``minimal inter-sampling interval'', MSI~\cite{Marchand:2013}) property.

An alternative event-triggered control algorithm, substantially relaxing the ISS condition~\eqref{eq.iss} and applicable, unlike~\cite{Marchand:2013,Marchand:2013IFAC}, to \emph{non-affine} systems, has been proposed in~\cite{SeuretPrieurMarchand:2013}. This approach requires the existence of a CLF that satisfies~\eqref{eq.iss} with $e=0$, i.e.
\be\label{eq.iss1}
%\begin{gathered}
\alpha_1(|x|)\le V(x)\le \alpha_2(|x|),\;\; V'(x)F(x,u(x))\le -\alpha_3(|x|).
%\end{gathered}
\ee
The events are triggered in a way providing that $V$ strictly decreases along any non-equilibrium trajectory
\be\label{eq.marchand-trig}
t_{n+1}=\inf\{t\ge t_n: V'(x(t))F(x(t),u_n)=-\mu(|x(t)|)\}.
\ee
Here $0<\mu(r)<\alpha_3(r)$ for any $r>0$ and $\mu$ is $\mathcal K_{\infty}$-function.
In general, this algorithm does not guarantee the dwell time positivity and may even lead to Zeno solutions~\cite{SeuretPrieurMarchand:2013}.

In this paper, we consider an algorithm similar in spirit to the algorithm from~\cite{SeuretPrieurMarchand:2013}.
Unlike~\cite{SeuretPrieurMarchand:2013}, in this paper we confine ourselves to CLFs that give \emph{exponentially}
stabilizing continuous-time controllers, which requires to modify the condition~\eqref{eq.iss1}.
In the case where such a CLF exists, we prove (under some natural assumptions) that
exponential convergence can also be provided by an event-triggered controller. Furthermore, such a controller
provides convergence rate arbitrary close to the rate of the continuous-time control and provides positive dwell time between consecutive switchings of the control input.
%Third, we show that this dwell time is in fact \emph{uniformly} positive over the set of trajectories, starting in a %compact region of the state space, and the knowledge of such a region
%allows to design self-triggered and periodic event-triggered controllers, which do not need constant monitoring of the %solution.
Unlike~\cite{Tabuada:2007,AntaTabuada:2010}, we do not assume that CLF satisfies the ISS condition~\eqref{eq.iss}. Unlike~\cite{Marchand:2013,Marchand:2013IFAC}, the affinity of the system
is not needed, and the convergence rate can be explicitly estimated. Unlike~\cite{SeuretPrieurMarchand:2013}, we prove the dwell time positivity (which, in particular, implies that all
solutions are non-Zeno).

\subsection{Exponential stabilization. Problem setup.}

Whereas the existence of CLF typically allows to find a stabilizing controller, it can potentially be unsatisfactory due to very slow convergence. Throughout this paper, we assume that the continuous-time CLF-based controller provides \emph{exponential} convergence rate; such a CLF is also called exponentially stabilizing, or ES-CLF~\cite{Ames:2014}.
Although finding of ES-CLF can be non-trivial, the inverse Lyapunov theorem~\cite{Khalil} implies that it usually exists in the vicinity of the equilibrium if
the system can be exponentially stabilized.
\begin{definition}
A function $V(x)$, satisfying~\eqref{eq.pos-def},  is said to
be an ES-CLF with exponent $\gamma>0$, if there exists a map $\u:\r^n\to U$, satisfying the conditions
\be\label{eq.inf-u-exp}
V'(x)F(x,\u(x))\le -\gamma V(x)\;\forall x,\quad F(0,\u(0))=0.
\ee
\end{definition}
\vskip 1mm
Note that the map $\u(\cdot)$ is not assumed to be continuous, so that the controller $u=\u(x)$ can be ``infeasible'',
that is, for some initial conditions the closed-loop system has no classical (Carathe\'odory's) solutions. For forward complete solutions,~\eqref{eq.inf-u-exp} implies that
\[
V(x(t))\le V(x(0))e^{-\gamma t}.
\]
Note that, in general, $V(x)$ need not be a monotone function of the norm $|x|$, so~\eqref{eq.inf-u-exp} does not imply~\eqref{eq.iss1}.

In this paper, we address the following fundamental question: does the existence of an ES-CLF allow to design an \emph{event-triggered} mechanism, also providing exponential convergence? In fact, we seek for event-triggered controllers whose convergence rates are \emph{arbitrarily close} to the rate of the continuous-time controller.

\textbf{Problem.} Given an ES-CLF $V$ with exponent $\gamma$ and a constant $\sigma\in (0,1)$, design an event-triggered controller, providing the exponential convergence with exponent $\sigma\gamma$
\be\label{eq.conv-rate1}
0\le V(x(t))\le V(x(0))e^{-\sigma\gamma t}.
\ee

\section{Event-triggered Controller Design}\label{sec.event}

Henceforth we suppose that an ES-CLF $V(x)$ and the corresponding feedback map $\u(x)$ from~\eqref{eq.inf-u-exp} are fixed.
By definition, for any $x$ we have $\u(x)\in U$. To simplify notation, denote
\be\label{eq.w}
\begin{gathered}
W(x,u)\dfb V'(x)F(x,u)\in\r,\quad x\in\r^n, u\in U.
\end{gathered}
\ee
The design of our event-triggered algorithm, to be discussed in what follows, provides
that
\be\label{eq.inf-u-exp-sigma}
\dot{V}(x(t))=W(x(t),u(t))\le-\sigma\gamma V(x(t))\quad\forall t\ge 0,
\ee
which evidently implies~\eqref{eq.conv-rate1}.

As usual in event-triggered control, the input $u(t)$ switches at sampling instants $t_0,t_1,\ldots$, whose sequence depends on the solution. At the initial instant $t_0=0$, compute the control input $u_0\dfb\u(x(t_0))$ and consider the solution corresponding to the input $u(t)=u_0,\, t\ge t_0$. If $V(x(t_0))=0$, then the system is already at the equilibrium $x(t_0)=0$ and remains there due to the assumption $F(0,\u(0))=0$. Otherwise, for $t$ sufficiently close to $t_0$ one has
\be\label{eq.aux1}
W(x(t),u_0)<-\sigma\gamma V(x(t))
\ee
since $W(x(t_0),u_0)\le -\gamma V(x(t_0))$ and $\sigma<1$. The next sampling instant $t_1$ is the \emph{first} time when~\eqref{eq.aux1} is violated; let $t_1=\infty$ if such an instant does not exist. If $t_1<\infty$, we repeat the procedure and compute the new control input $u_1=\u(x(t_1))$, which remains unchanged till the next sampling instant $t_2$.
If $V(x(t_1))=0$, then the system stays at the equilibrium under the control input
$u(t)\equiv u_1,\, t\ge t_1$ and we put $t_2=\infty$. Otherwise, for $t$ sufficiently close to $t_1$ the following inequality holds
\be\label{eq.aux1+}
W(x(t),u_1)<-\sigma\gamma V(x(t)).
\ee
Let the next sampling instant $t_2$ be the first time $t>t_1$ when~\eqref{eq.aux1+} is violated and $t_2=\infty$ if such an instant does not exist. Iterating this procedure,
the sequence of instants $\{t_n\}$ is constructed in a way that the control $u(t)=u_n\dfb\u(x(t_n))$ for $t\in[t_n,t_{n+1})$ satisfies~\eqref{eq.inf-u-exp-sigma}.
If $V(x(t_n))>0$, $t_{n+1}$ is the first time $t>t_n$ when
\be\label{eq.inf-u-sigma-eq-n}
W(x(t),u_n)=-\sigma\gamma V(x(t)).
\ee
The sequence of sampling instants terminates if $V(x(t_n))=0$ or~\eqref{eq.inf-u-sigma-eq-n} does not hold at any $t>t_n$, in this case we formally define $t_{n+1}=\infty$ and the control remains constant $u(t)\equiv u_n$ for $t>t_n$.

The procedure just described can be written mathematically as
\be\label{eq.alg1}
\begin{gathered}
u(t)=\u(x(t_n))\;\;\forall t\in [t_n,t_{n+1}),\quad t_0=0,\\
t_{n+1}=\begin{cases}\inf\left\{t>t_n: \eqref{eq.inf-u-sigma-eq-n}\;\text{holds} \right\},\,&V(x(t_n))> 0,\\
\infty,\,&V(x(t_n))=0.
\end{cases}
\end{gathered}
\ee
(for brevity, we assume that $\inf\emptyset\dfb +\infty$).

Note that implementation of Algorithm~\eqref{eq.alg1} assumes implicitly that a constructive procedure (e.g. optimization-based) is available to compute
$\u(x(t_n))$ at each sampling instant, however, it \emph{does not} require any closed-form analytic expression for $\u(x)$.
%For instance, if the system is affine (i.e. $F(x,u)=f(x)+g(x)u$) and $U$ is a closed \emph{convex} set, one of controls $\u(x)$ can be found by solving the optimization problem
%\be\label{eq.optimum}
%\begin{gathered}
%|u|^2\to\min\\
%\text{subject to}\quad W(x,u)\le-\gamma V(x),\quad u\in U.
%\end{gathered}
%\ee
%By definition of ES-CLF, the set of admissible values $u$ is not empty. Since this set is closed and convex, the
%minimum in~\eqref{eq.optimum} exists.

To assure the practical applicability of the algorithm~\eqref{eq.alg1}, one has to prove that the solution of the closed-loop system is unique and forward complete, addressing the following two problems. First, one has to establish the solution's existence and uniqueness between two sampling instants. In particular, one has to show that the event~\eqref{eq.inf-u-sigma-eq-n} is detected earlier than the solution runs away to infinity.
Second, one has to show the absence of Zeno trajectories, for which the sequence $t_n$ converges to a finite limit.
\begin{definition}
A solution to the closed-loop system~\eqref{eq.syst},\eqref{eq.alg1} is said to be \emph{Zeno}, or exhibit \emph{Zeno behavior} if the sequence of sampling instants is infinite and $\lim\limits_{n\to\infty}t_n=\sup\limits_{n\ge 0}t_n<\infty$.
\end{definition}

Notice that even for non-Zeno trajectories it may happen that $t_{n+1}-t_n\to 0$ as $n\to\infty$, which makes it problematic to implement the algorithm on a real-time platform. Thus we are primarily interested in the more restrictive condition of the \emph{dwell-time} positivity.
\begin{definition}
The value $\mathfrak{T}=\mathfrak{T}(x_0)=\inf\limits_{n\ge 0}(t_{n+1}-t_n)$ is called the \emph{dwell-time} of the solution. Algorithm~\eqref{eq.alg1} provides
positive dwell time if $\mathfrak{T}(x_0)>0$. We say that the algorithm provides \emph{locally uniformly} positive dwell time if the function $\mathfrak{T}$ is
uniformly positive on any compact set $K$, e.g. $\inf\limits_{x_0\in \mathcal K}\mathfrak{T}(x_0)>0$.
\end{definition}

In this paper, we establish criteria for local uniform (called sometimes ``semi-uniform''~\cite{Marchand:2013}) positivity of the dwell time $\mathfrak{T}$.

\subsection{The inter-sampling behavior of solutions}

To examine the solution's behavior between two sampling instants, consider the following Cauchy problem
\be\label{eq.cauchy}
\dot \xi(t)=F(\xi(t),u_*),\; \xi(0)=\xi_0,\;t\ge 0,
\ee
where $u_*\in U$. To provide the unique solvability of~\eqref{eq.cauchy}, henceforth the following non-restrictive assumption is adopted.
\begin{assumption}\label{ass.contin}
For $u_*\in U$, the map $F(\cdot,u_*)$ is locally Lipschitz, and hence the function $W(\cdot,u_*):\r^n\to\r$ is continuous.
\end{assumption}

\begin{proposition}\label{prop.tech}
Under Assumption~\ref{ass.contin}, the Cauchy problem~\eqref{eq.cauchy} has the unique solution $\xi(t)=\xi(t|\xi_0,u_*)$,
whose maximal interval of existence either contains a point $t$ such that $W(\xi(t),u(t))>-\sigma\gamma V(\xi(t))$ or is infinite (the solution is forward complete).
\end{proposition}
\begin{proof}
The first statement follows from the Picard-Lindel\"of theorem~\cite{Khalil}. It remains to prove that the solution cannot grow infinite while $W(\xi(t),u(t))\le -\sigma\gamma V(\xi(t))$. Indeed, the latter condition implies that $\dot{V}(\xi(t))\le -\sigma\gamma V(\xi(t))\le 0$, and thus $V(\xi(t))\le V(\xi_0)$. Recalling that $V$ is proper, one obtains boundedness of $\xi(t)$.
\end{proof}
\begin{corollary}\label{cor.unique1}
Under Assumption~\ref{ass.contin}, $x(t)=\xi(t-t_+|x_+,u_*)$ is the unique solution to the Cauchy problem
\be\label{eq.cauchy+}
\dot x(t)=F(x(t),u_*),\; x(t_+)=x_+,\;t\ge t_+,
\ee
where $x_+\in\r^n$ and $u_*\in U$. If $x_+=0$ and $u_*=\u(0)$, then $x(t)\equiv 0$.
\end{corollary}

Applying Corollary~\ref{cor.unique1} to $t_+=t_n$, $x_+=x(t_n)$ and $u_*=\u(x(t_n))$, one shows that the sequence of sampling instants $t_n$ in~\eqref{eq.alg1} is well defined, and
the instant $t_{n+1}$ depends only on $t_n$ and $x(t_n)$.
\begin{corollary}\label{cor.t-n-exist}
Let Assumption~\ref{ass.contin} hold. For each sampling instant $t_n$, the solution to the Cauchy problem
\be\label{eq.cauchy-n}
\dot x(t)=F(x(t),u_n),\quad u_n=\u(x(t_n)),\,t\ge t_n
\ee
either satisfies the triggering condition~\eqref{eq.inf-u-sigma-eq-n} at some time $t>t_n$
(that is, $t_{n+1}<\infty$) or is forward complete and satisfies the inequality $W(x(t),u_n)<-\sigma\gamma V(x(t))\,\forall t\ge t_n$.
\end{corollary}
\begin{remark}\label{eq.non-increase}
By construction of the sampling instants, the inequality~\eqref{eq.inf-u-exp-sigma} holds between them and, in particular,
the CLF $V(x(t))$ is non-increasing along each trajectory.
\end{remark}

\subsection{Dwell time positivity}

In this subsection, we formulate our main result, namely, the criterion of dwell time positivity in Algorithm~\eqref{eq.alg1}. This criterion relies on additional assumptions.

For any $x_*\in\r^n$ and $\mathcal K\subset\r^n$, denote
\be\label{eq.b}
B(x_*)\dfb\{x: V(x)\le V(x_*)\},\;\; B(\mathcal K)\dfb \bigcup\limits_{x_*\in \mathcal K}B(x_*).
\ee
The set $B(\mathcal K)$ is bounded for any bounded set $\mathcal K$ since $V(x)$ is supposed to be continuous and radially unbounded
\[
B(\mathcal K)\subseteq\{x: V(x)\le \sup_{x_*\in \mathcal K} V(x_*)\}.
\]
Accordingly to Assumption~\ref{ass.contin}, the following supremum is finite
\be\label{eq.kappa}
\varrho(x_*)\dfb\sup\limits_{\substack{x_1,x_2\in B(x_*)\\x_1\ne x_2}} \frac{|F(x_1,\u(x_*))-F(x_2,\u(x_*))|}{|x_2-x_1|}<\infty\\
\ee
for any $x_*$ (in the case where $x_*=0$ and $B(x_*)=\{0\}$, let $\varrho(x_*)\dfb 0$). We adopt a stronger version of Assumption~\ref{ass.contin}.
\begin{assumption}\label{ass.F}
The Lipschitz constant $\varrho(x_*)$ in~\eqref{eq.kappa} is a \emph{locally bounded} function of $x_*$ (that is, $\rho$ is bounded on any compact).
\end{assumption}

Assumption~\ref{ass.F} holds, for instance, if the mapping $\u$ is locally bounded and the derivative $F'_x(x,u)$ is continuous in $x$ and $u$.
The next assumption is a stronger version of CLF's smoothness.
\begin{assumption}\label{ass.gradient}
The function $V'(x)$ is locally Lipschitz.
\end{assumption}

Along with $\varrho(\cdot)$, we introduce the Lipschitz constant of the gradient $V'$ on the compact set $B(x_*)$ as follows
\be\label{eq.nu}
\nu(x_*)\dfb\sup\limits_{\substack{x_1,x_2\in B(x_*)\\x_1\ne x_2}} \frac{|V'(x_1)-V'(x_2)|}{|x_2-x_1|},\,\nu(0)=0.
\ee
Since for any compact $\mathcal{K}$ the set $B(\mathcal K)$ is bounded, one has
\[
\sup_{x_*\in \mathcal K}\nu(x_*)\le \sup\limits_{\substack{x_1,x_2\in B(\mathcal K)\\x_1\ne x_2}} \frac{|V'(x_1)-V'(x_2)|}{|x_2-x_1|}<\infty.
\]
Assumption~\ref{ass.gradient} thus implies that $\nu(\cdot)$ from~\eqref{eq.nu} is \emph{locally bounded}.

Finally, we adopt an assumption that allows to establish the relation between the convergence rates of the ES-CLF $V(x(t))$ under the continuous-time control $\u=\u(x)$ and the solution $x(t)$. Notice that~\eqref{eq.inf-u-exp} gives no information about the speed of the solution's convergence since $\dot V(x)=V'(x)\dot x(t)$ depends only on the velocity's $\dot x(t)$ projection on the gradient vector $V'(x)$, whereas the transversal components can be arbitrary. These transversal dynamics can potentially lead to very slow and ``non-smooth'' convergence, e.g., the velocity $\dot x=F(x,\u(x))$ can be unbounded as $x\to 0$. Our final assumption excludes this pathological behavior. For brevity, let
\[
\bar F(x)\dfb F(x,\u(x)).
\]
\begin{assumption}\label{ass.non-degen}
The ES-CLF $V(x)$ and the corresponding controller $\u(x)$ satisfy the following properties:
\be\label{eq.non-degen}
\begin{gathered}
|\bar F(x)|\le M_1(x)|V'(x)|\quad\forall x\in\r^n,\\
\cos\theta(x)\le -M_2(x)\quad\forall x\in\r^n\setminus\{0\}.
\end{gathered}
\ee
Here $\theta(x)$ stands for the angle between the vectors $\bar F(x)$ and $V'(x)$ (Fig.~\ref{fig.1}), $M_1$ is locally bounded, and $M_2$ is locally strictly positive\footnote{In other words, on any compact set the function $M_1$ is bounded and the function $M_2$ is uniformly strictly positive.}.
\end{assumption}

The inequalities~\eqref{eq.non-degen} imply that the solution does not oscillate near the equilibrium since $|\bar F(x)|\to 0$ as $|x|\to 0$, and the angle between the vectors\footnote{The inequality~\eqref{eq.inf-u-exp} implies that both vectors are non-zero unless $x\ne 0$} $\dot x=\bar F(x)$ and $V'(x)$ remains strictly obtuse as $x\to 0$, i.e. the flow is not transversal to the CLF's gradient.
\begin{figure}[h]
\center
\includegraphics[height=3cm]{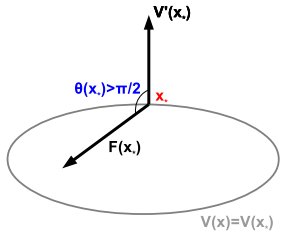}
\caption{Illustration to Assumption~\ref{ass.non-degen}: the angle $\theta(x_*)$}\label{fig.1}
\end{figure}

Assumption~\ref{ass.non-degen} can be reformulated as follows.
\begin{lemma}\label{lem.tech}
For an ES-CLF $V$, Assumption~\ref{ass.non-degen} holds if and only if
a locally bounded function $M(x)>0$ exists such that
\be\label{eq.non-degen1}
|V'(x)|\,|\bar F(x)|+|\bar F(x)|^2\leq M(x)|V'(x)\bar F(x)|\quad\forall x.
\ee
\end{lemma}
\begin{proof}
To prove the ``only if'' part, notice that $|V'(x)\bar F(x)|=|\cos\theta(x)||V'(x)|\,|\bar F(x)|\ge M_2(x)|V'(x)|\,|\bar F(x)|$. Therefore $|\bar F(x)|^2\le M_1(x)|V'(x)|\,|\bar F(x)|\le M_1(x)/M_2(x)|V'(x)\bar F(x)|$ and~\eqref{eq.non-degen1} holds for $M(x)=M_1(x)+M_1(x)/M_2(x)$. To prove ``if'' part, note that
\[
M(x)\cos\theta(x)=\frac{M(x)V'(x)\bar F(x)}{|V'(x)|\,|\bar F(x)|}\overset{\eqref{eq.non-degen1},\eqref{eq.inf-u-exp}}{\le} -1
\]
and $|\bar F(x)|^2\leq M(x)|V'(x)\bar F(x)|\le M(x)|V'(x)|\,|\bar F(x)|$. Hence $|\bar F(x)|\leq M(x)|V'(x)|$
and~\eqref{eq.non-degen} holds with $M_1=M$ and $M_2=1/M$.
\end{proof}

Assumption~\ref{ass.non-degen} restricts the solution to approach the equilibrium ``smoothly'' in the sense that the state $x(t)$ cannot change much faster than the CLF decreases along it.
Note that the definition of ES-CLF~\eqref{eq.inf-u-exp} implies the following
``relaxed'' version of this assumption. First, $V'(x)=0$ implies that $\gamma V(x)\le V'(x)\bar F(x)=0$, that is, $x=0$ and thus $\bar F(x)=0$, in other words, $|\bar F(x)|/|V'(x)|<\infty$ at any $x\ne 0$. Second, the angle between $V'(x)$ and $\bar F(x)$ has to be obtuse $\cos\theta(x)<0$ unless $x=0$.
It is convenient to verify  Assumption~\ref{ass.non-degen} and the condition~\eqref{eq.inf-u-exp} simultaneously since both of these conditions involve $V'(x)$ and $\bar F(x)=F(x,u(x))$.

We now formulate a key technical lemma which allows to establish the criterion of dwell time positivity in Algorithm~\eqref{eq.alg1}.
This lemma, proved in Appendix, entails that the time $t_{n+1}-t_n$ elapsed between two consecutive events cannot be smaller than $\tau_{\sigma}(x(t_n))$, where $\tau_{\sigma}(\cdot)$ is
a locally uniformly positive function (depending only on $\rho(\cdot),\nu(\cdot),M(\cdot)$ and $\sigma$).
Consider again the solution $\xi(t)=\xi(t|x_*,\u(x_*))$ to the system~\eqref{eq.cauchy} with $\xi_0=x_*$ and $u_*=\u(x_*)$.
\begin{lemma}
\label{lem.key-lemma}
Let the system~\eqref{eq.syst} and the ES-CLF $V(x)$ satisfy Assumptions~\ref{ass.F}-\ref{ass.non-degen}. Then for any $\sigma\in(0,1)$ there exists a function $\tau_{\sigma}:\r^n\to (0,\infty)$, featured by the following properties:
\begin{enumerate}
\item $\tau_{\sigma}(\cdot)$ is uniformly strictly positive on any compact set;
\item for any $x_*\ne 0$ the function $\xi(t)=\xi(t|x_*,\u(x_*))$ is well-defined on $[0,\tau_{\sigma}(x_*)]$ and, furthermore,
\be\label{eq.w-ineq-main}
W(\xi(t),\u(x_*))<-\sigma\gamma V(\xi(t))\quad\forall t\in [0,\tau_{\sigma}(x_*)).
\ee
\end{enumerate}
If the functions $\varrho(x_*)$, $\nu(x_*)$ and $M(x_*)$) are \emph{globally} bounded, then $\tau_{\sigma}(x_*)$ is \emph{uniformly} strictly positive on $\r^n$.
\end{lemma}

The proof of Lemma~\ref{lem.key-lemma} will be given in Appendix.
Note that Algorithm~\eqref{eq.alg1} \emph{does not} employ the functions $\varrho(x_*)$, $\mu(x_*)$ and $M(x_*)$ in any way; they influence only the dwell time estimate $\tau_{\sigma}(\cdot)$.
The explicit formula for $\tau_{\sigma}(x_*)$, given in Appendix, shows that $\tau_{\sigma}$ is non-increasing in $\sigma$, being maximal for $\sigma=0$
and vanishing as $\sigma\to 1$. Recalling that $\sigma$ regulates the convergence speed of the algorithm, one can notice
that the price paid for the fast convergence is the small dwell time between the consecutive events (or, equivalently,
large number of events per unit of time). Our main result is the following criterion of dwell time positivity.
\begin{theorem}\label{thm.dwell}
Let the assumptions of Lemma~\ref{lem.key-lemma} hold. Then Algorithm~\eqref{eq.alg1} provides locally uniformly positive dwell time
\be\label{eq.tau-min}
\mathfrak{T}(x_0)\ge \tau_{\sigma,min}(x_0)\dfb\inf_{x\in B(x_0)}\tau_{\sigma}(x).
\ee
Here $\tau_{\sigma}(x)$ stands for the function from Lemma~\ref{lem.key-lemma}.
\end{theorem}
\begin{proof}
Notice first that the function $\tau_{\sigma,min}$ from~\eqref{eq.tau-min} is uniformly strictly positive on any compact set $K\subseteq\r^n$ since
\[
\inf_{x_0\in\mathcal K}\tau_{\sigma,min}(x_0)=\inf_{x\in B(\mathcal K)}\tau_{\sigma}(x)>0,
\]
$B(K)$ is bounded and thus $\tau_{\sigma}$ is strictly positive on $B(K)$.
Remark~\eqref{eq.non-increase} implies that each set $B(x_*)$ is forward invariant, in particular, the solution starting at $x(0)=x_0$ remains
in $B(x_0)$. If an event occurs at $t=t_n$, then Lemma~\ref{lem.key-lemma} applied to $x_*=x(t_n)$ entails that the next event cannot occur earlier than at
$t=t_n+\tau_{\sigma}(x(t_n))\ge t_n+\tau_{\sigma,min}(x_0)$, that is, $t_{n+1}-t_n\ge\tau_{\sigma,min}(x_0)$ for any $n$.
\end{proof}

Proposition~\ref{prop.tech} and Theorem~\ref{thm.dwell} imply, in particular, that algorithm~\eqref{eq.alg1} is feasible in
the sense that for any $x(0)$, the closed-loop system has the unique solution, which is forward complete.

\subsection{Extensions}

We now consider two important extensions of the main result, dealing with non-exponential stability and safety-critical systems.

\subsubsection{Non-exponential convergence}
Our algorithm~\eqref{eq.alg1} can be easily modified to cope with many CLFs that do not provide exponential convergence. For instance, replacing~\eqref{eq.inf-u-exp} by the inequality
\be\label{eq.inf-u-p}
V'(x)F(x,\u(x))\le -\gamma V(x)^p\;\forall x,\quad F(0,\u(0))=0,
\ee
with $p>1$ and $\gamma>0$, one has $V(x(t))=O\left(t^{\frac{1}{1-p}}\right)$ as $t\to\infty$ since
\[
V(x(t))\le \left(V(x(0))+\gamma(p-1)t\right)^{\frac{1}{1-p}}.
\]
The arbitrarily close convergence rate is provided by the modification of algorithm~\eqref{eq.alg1}, where~\eqref{eq.inf-u-sigma-eq-n} is replaced by
\[
W(x(t),u_n)=-\sigma\gamma V(x(t))^p.
\]
Instead of~\eqref{eq.inf-u-exp-sigma}, such an algorithm provides the condition
\ben
V'(x(t))F(x(t),u(t))\le -\sigma\gamma V(x(t))^p,
\een
giving an explicit estimate of the convergence rate
\be\label{eq.conv-rate-p}
V(x(t))\le \left(V(x(0))+\sigma\gamma(p-1)t\right)^{\frac{1}{1-p}}.
\ee
A closer analysis of the proofs reveals that all statements from Subsect.~3.2, including the dwell time positivity criterion from Theorem~\ref{thm.dwell}, retain their validity for such a modified algorithm.

\subsubsection{Safety-critical control}
For many safety-critical systems, such as e.g. autonomous robots, smart factories and power grids, safety has to be
provided by the control design. Often the requirement of safety can be mathematically described as avoiding of some ``dangerous'' set $\mathcal D$ by the solution $x(t)\not\in\mathcal D$. As has been demonstrated in~\cite{RomdlonyJayawardhana:2016},
in many situation the stabilization problem with this additional restriction can be solved by using \emph{control Lyapunov-Barrier functions} (CLBF). We do not consider here the general definition of CLBF from~\cite{RomdlonyJayawardhana:2016} and only formulate a simple result, concerned with safe stabilization.
As usual, $Int \mathcal D$ denotes the interior of the set $\mathcal D$, and $\partial D$ stands for its boundary.
\begin{lemma}
Let $\mathcal D\subset\r^n\setminus\{0\}$ stand for the \emph{closed} set of ``dangerous'' states, we assume that
$\mathcal D=\overline{Int\mathcal D}$. Suppose that a CLF\,$V(x)$ also serves as a barrier certificate in the sense that for any $\xi\in\partial\mathcal D$ one has $V(\xi)\ge v_*>0$. Then for any $x(0)\not\in \mathcal D$ such that $V(x(0))<v_*$, the event-triggered algorithm~\eqref{eq.alg1} provides safety $x(t)\not\in\mathcal D$.
\end{lemma}
\begin{proof}
Indeed, the design of the algorithm provides that~\eqref{eq.inf-u-exp-sigma} holds along any solution, in particular,
$V(x(t))<v_*$. Therefore, the solution cannot cross the boundary of the set $\mathcal D$.
\end{proof}

In particular, if the assumptions of Theorem~\ref{thm.dwell} are valid, the algorithm~\eqref{eq.alg1} provides exponential event-triggered stabilization with guaranteed safety whenever $x(0)\not\in \mathcal D$ and $V(x(0))<v_*$.
Obviously, the exponential convergence~\eqref{eq.inf-u-exp} can be replaced by~\eqref{eq.inf-u-p}.

\section{Examples}

We illustrate algorithm~\eqref{eq.alg1} by considering two examples.

\subsection{Example~1. Event-triggered backstepping}

Event-triggered control proves to be an important tool in design of cooperative control algorithms for automated driving,
where communication between the vehicles is seriously restricted by the wireless network bandwidth~\cite{DolkBorgersHeemels:2017,DolkPloegHeemels:2017}.
In this subsection, we consider a simplified problem of two vehicle platoons merging~\cite{Chien:1995}. Assume that the lead platoon (Fig.~\ref{fig.2}) travels at constant speed $v_0>0$, an algorithm is wanted allowing the follower (trail)
platoon to merge safely with it. Denoting the velocity of the trail platoon's leader by $v(t)$ and its
distance to the lead platoon (Fig.~\ref{fig.2}) by $d(t)$, the merging goal can be formulated as follows~\cite{Chien:1995}
\be\label{eq.merge}
d(t)-d_0\xrightarrow[t\to\infty]{} 0,\quad v(t)-v_0\xrightarrow[t\to\infty]{} 0,
\ee
where $d_0$ is the desired safe inter-vehicle distance. In general, more complicated
speed control policies are required~\cite{AlvarezHorotwitz:1999}, ensuring safety in the case where the lead platoon applies emergency braking. Such merging algorithms are beyond the scope of this paper.
\begin{figure}[h]
\center
\includegraphics[width=\columnwidth]{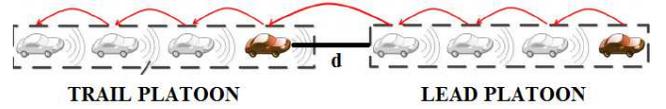}
\caption{Two platoons merging}\label{fig.2}
\end{figure}

Our goal is to design the algorithm for the leading vehicle of the trail platoon, providing the control goal~\eqref{eq.merge}. Unlike~\cite{Chien:1995}, dealing with highly nonlinear controllers for the throttle and braking systems of the vehicle, we suppose that the vehicle's longitudinal dynamics can be approximated~\cite{StotskyChien:1995} by the equation
\be\label{eq.vehi}
\tau(v)\dot a(t)+a(t)=u(t).
\ee
Here $a(t)=\dot v(t)$ is the leading vehicle's actual acceleration, whereas $u(t)$ can be treated as the desired acceleration command.
Note that, in general, the system~\eqref{eq.vehi} is \emph{nonlinear} due to the presence of function $\tau(v)$, depending on the dynamics of the servo-loop and characterizing time lag
between the commanded and actual accelerations. We suppose that $\tau(v)$ is known, the trail platoon's leader measures $d(t),v(t),a(t)$ and is aware of the lead platoon's speed $v_0$.

To design an ES-CLF for this stabilization problem, we use the well-known backstepping procedure~\cite{KrsticKokotovicBook,Khalil}. Choosing a parameter $k>1$, we introduce the new state variables $x_1,x_2,x_3$ as follows
\ben
\begin{gathered}
x_1(t)\dfb d(t)-d_0,\quad
x_2(t)\dfb \dot x_1(t)+kx_1(t)=(v_0-v(t))+kx_1(t)\\
x_3(t)\dfb \dot x_2(t)+kx_2(t)= -a(t)+2k(v_0-v(t))+k^2x_1(t).
\end{gathered}
\een

By noticing that $v_0-v(t)=x_2-kx_1$ and $a(t)=2kx_2(t)-k^2x_1(t)-x_3(t)$, the equations~\eqref{eq.vehi} are rewritten as follows
\be\label{eq.vehi1}
\begin{aligned}
\dot x_1&=x_2-kx_1\\
\dot x_2&=x_3-kx_2\\
\dot x_3&=k^2[x_2-kx_1]+\\
&+[\tau(v)^{-1}-2k](2kx_2-k^2x_1-x_3)-\tau(v)^{-1}u\\
v&=v_0-(x_2-kx_1).
\end{aligned}
\ee
The backstepping procedure implies that $V(x)=\frac{1}{2}(x_1^2+x_2^2+x_3^2)$ is the ES-CLF for the system~\eqref{eq.vehi1}, associated with the controller
\[
\begin{gathered}
\u(x)\dfb\tau(v)k^2[x_2-kx_1]+[1-2k\tau(v)](2kx_2-k^2x_1-x_3)-x_1+kx_3.
\end{gathered}
\]
A straightforward computation shows that
\[
\begin{aligned}
F(x,\u(x))&=(x_2-kx_1,x_3-kx_2,x_1-kx_3)^{\top},\\
V'(x)F(x,\u(x))&=-2(k-1)V(x)-\\&-\frac{1}{2}[(x_1-x_2)^2+(x_1-x_3)^2+(x_2-x_3)^2],
\end{aligned}
\]
entailing~\eqref{eq.inf-u-exp} with $\gamma=2(k-1)$. It can be easily shown that all assumptions of Theorem~\ref{thm.dwell} (in particular,~\eqref{eq.non-degen1}) hold. The algorithm~\eqref{eq.alg1} is an event-triggered controller for platoons' merging.

In Fig.~\ref{Fig.ex1}, we simulate the behavior of the algorithm~\eqref{eq.alg1} with $\sigma=0.9$, choosing $k=1.005$ and $\tau=0.5s$. The initial condition corresponds to the situation where $x_1(0)=d(t)-d_0=10$, $x_2(0)=kx_1(0)$, $x_3(0)=k^2x_1(0)$. In other words, at the initial time the trail platoon has the same speed as the lead platoon $v(0)=v_0$, $a(0)=0$, whereas the distance to the lead platoon exceeds the desired reference value by 10m.
One may notice that the maneuver of the trail platoon's leader includes a short period of ``harsh'' braking, which causes discomfort for human occupants of the vehicle. Vehicle platooning under realistic safety and comfort constraints is a non-trivial problem, which is beyond the scope of this paper.
\begin{figure}[h]
\begin{subfigure}[t]{0.75\columnwidth}
\center
\includegraphics[width=\columnwidth]{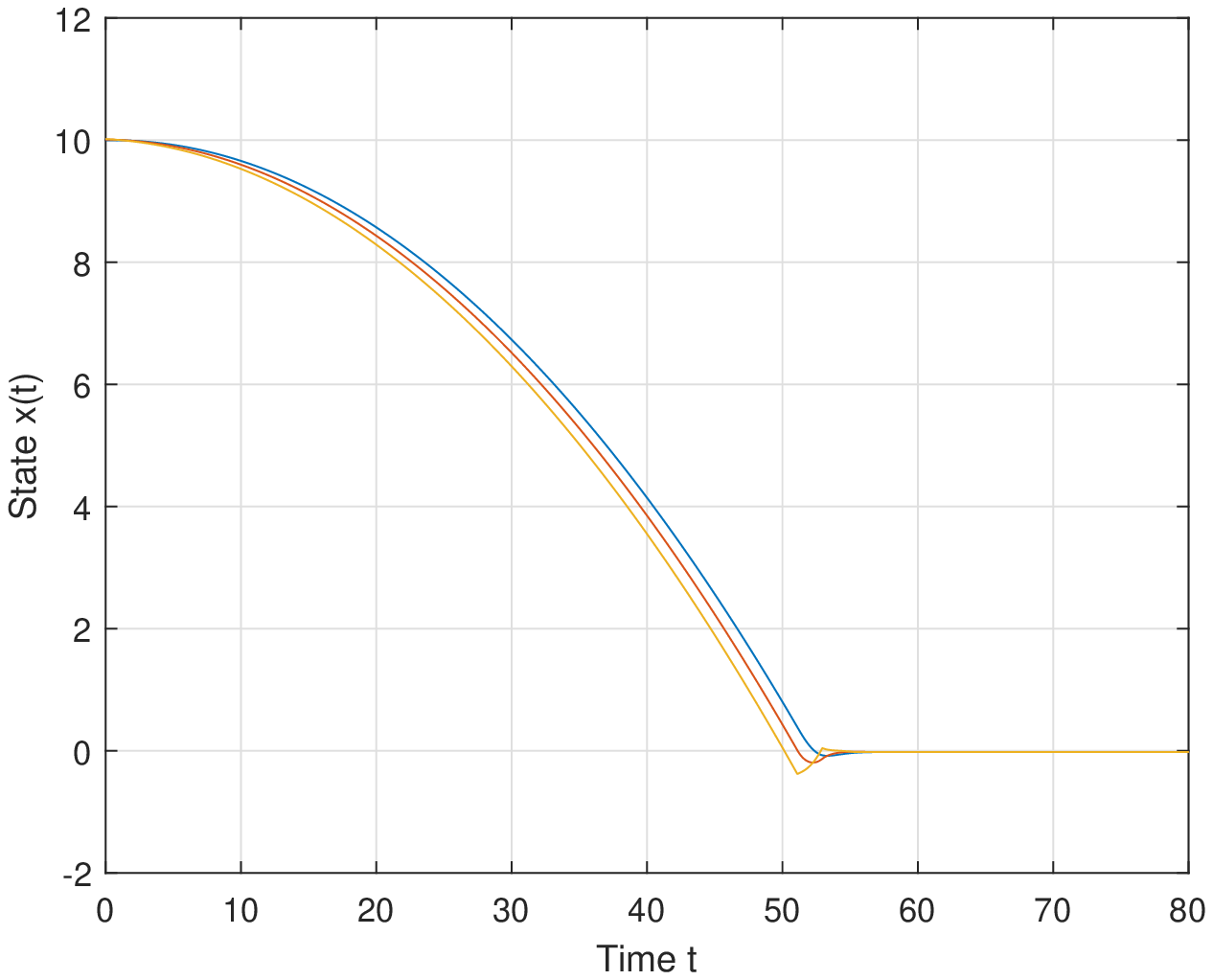}
\end{subfigure}
\begin{subfigure}[t]{0.75\columnwidth}
\center
\includegraphics[width=\columnwidth]{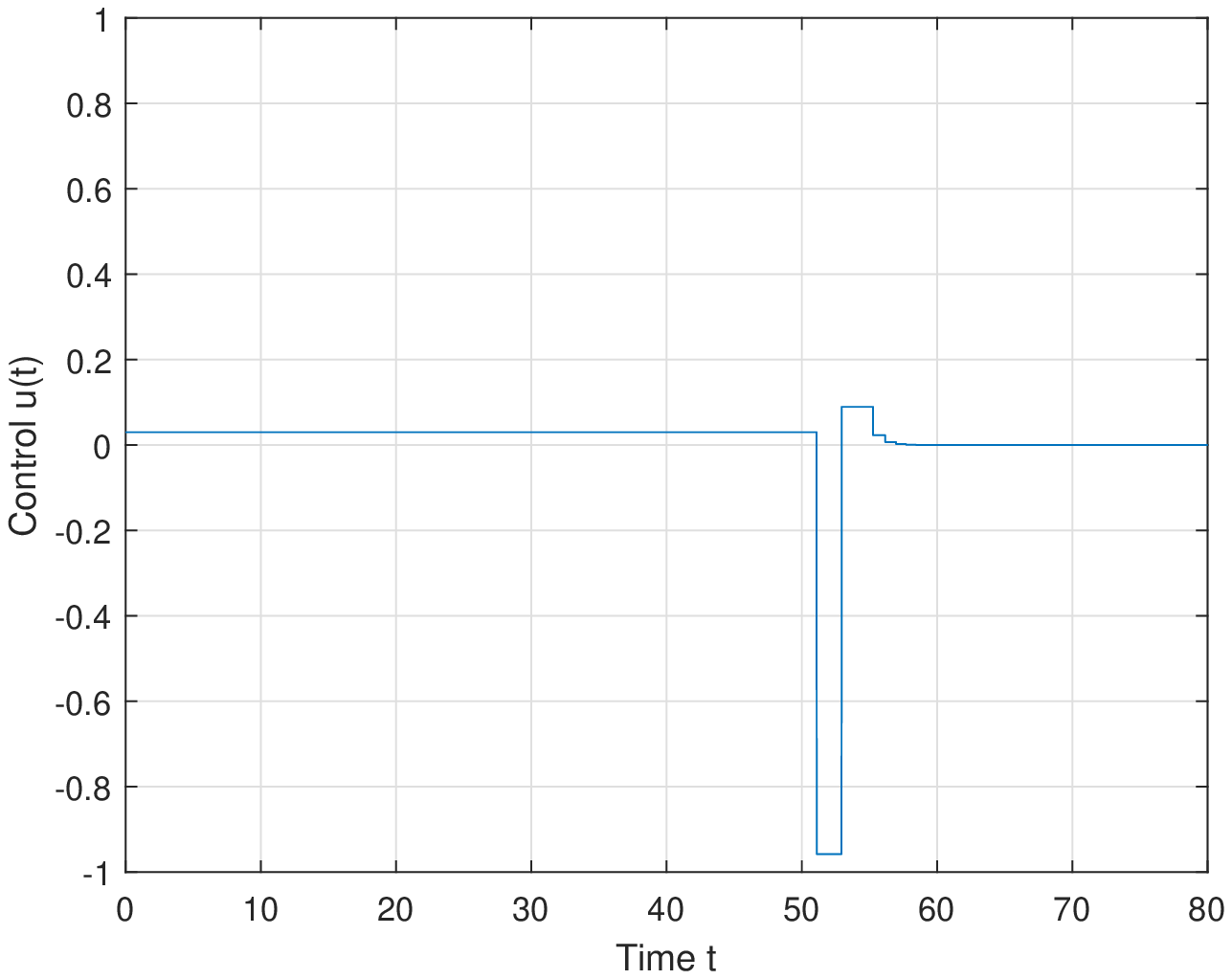}
\end{subfigure}
\begin{subfigure}[t]{0.75\columnwidth}
\center
\includegraphics[width=\columnwidth]{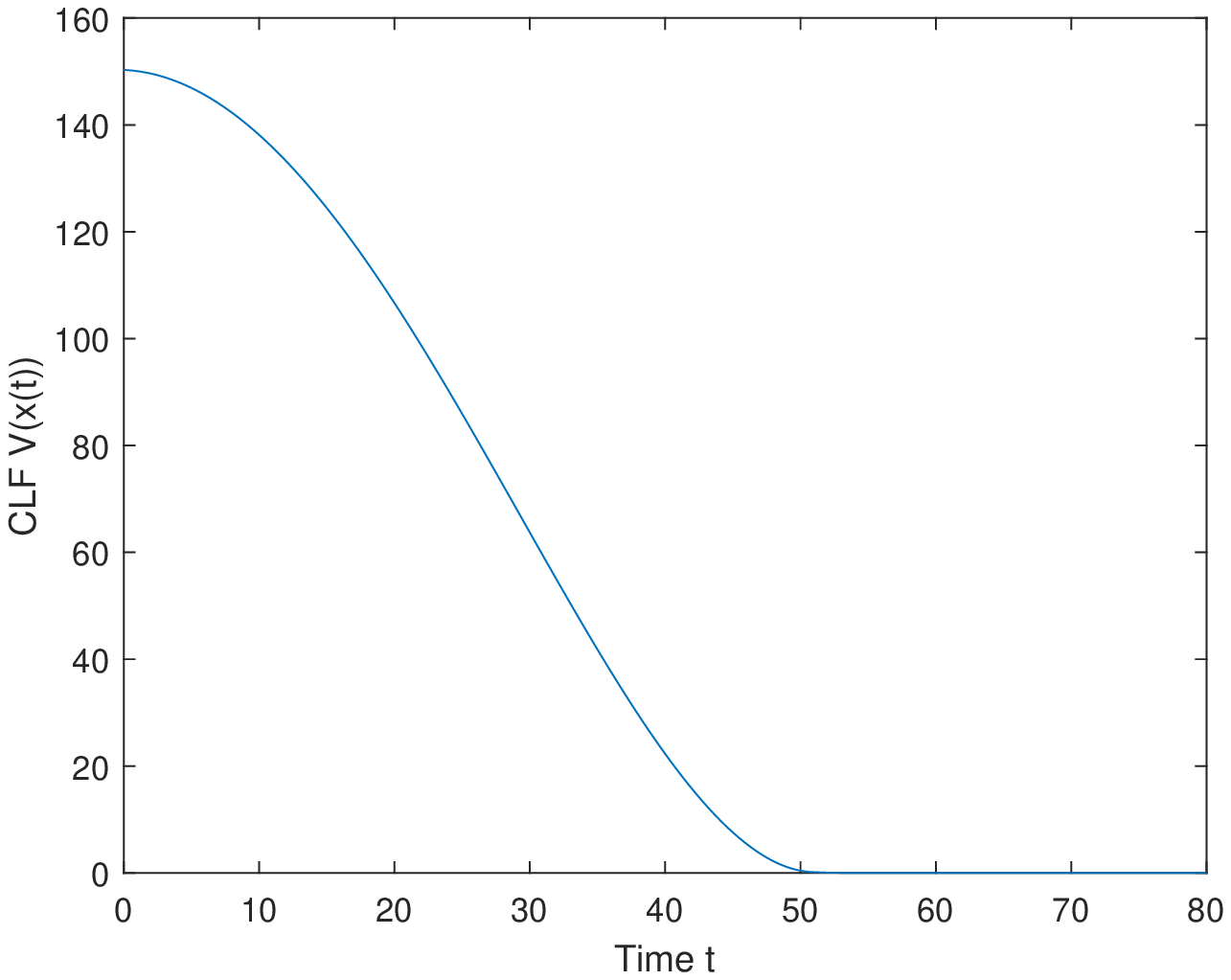}
\end{subfigure}
\caption{Event-triggered stabilization of system~\eqref{eq.vehi1}}\label{Fig.ex1}
\end{figure}

\subsection{Example~2. Non-exponential stabilization}

Our second example is borrowed from~\cite{AntaTabuada:2008} and deals with a two-dimensional homogeneous system
\be\label{eq.syst-anta}
\begin{gathered}
\dot x_1=-x_1^3+x_1x_2^2,\\
\dot x_2=x_1x_2^2+u-x_1^2x_2
\end{gathered}
\ee

The quadratic form $V(x)=\frac{1}{2}[x_1^2+x_2^2]$ is not an ES-CLF, however, it satisfies~\eqref{eq.inf-u-p}
with $p=2$, where $\u(x)=-x_2^3-x_1x_2^2$ since
\[
V'(x)F(x,\u(x))=-x_1^4-x_2^4\le -V^2/2.
\]
According to~\eqref{eq.conv-rate-p}, the event-triggered algorithm~\eqref{eq.alg1} provides the stabilization with the convergence rate
\be\label{eq.conv-rate-exam}
V(x(t))\le \left[V(x(0))+\sigma t/2\right]^{-1}.
\ee

To compare our algorithm with the one reported in~\cite{Marchand:2013}, we simulate the behavior of the system for $x_1(0)=0.1, x_2(0)=0.4$, choosing $\sigma=0.9$. The results of numerical simulation (Fig.~\ref{Fig.ex2}) are very similar to those presented in~\cite{Marchand:2013}. Although the convergence of the solution is slow ($V(x(t))=O(t^{-1})$, and hence $|x(t)|=O(t^{-1/2})$), its second component (and thus also the control input) converges very fast. During the first $200$s,
only two events are detected at times $t_0=0$ and $t_1\approx 5.26$, after the second event the control input is $u(t)\approx -6\cdot 10^{-7}$.
Unlike the controller from~\cite{Marchand:2013}, based on the Sontag formula~\cite{Sontag:1989}, our algorithm provides the explicit estimate of the solution's convergence rate~\eqref{eq.conv-rate-p}.
\begin{figure}[h]
\begin{subfigure}[t]{0.75\columnwidth}
\center
\includegraphics[width=\columnwidth]{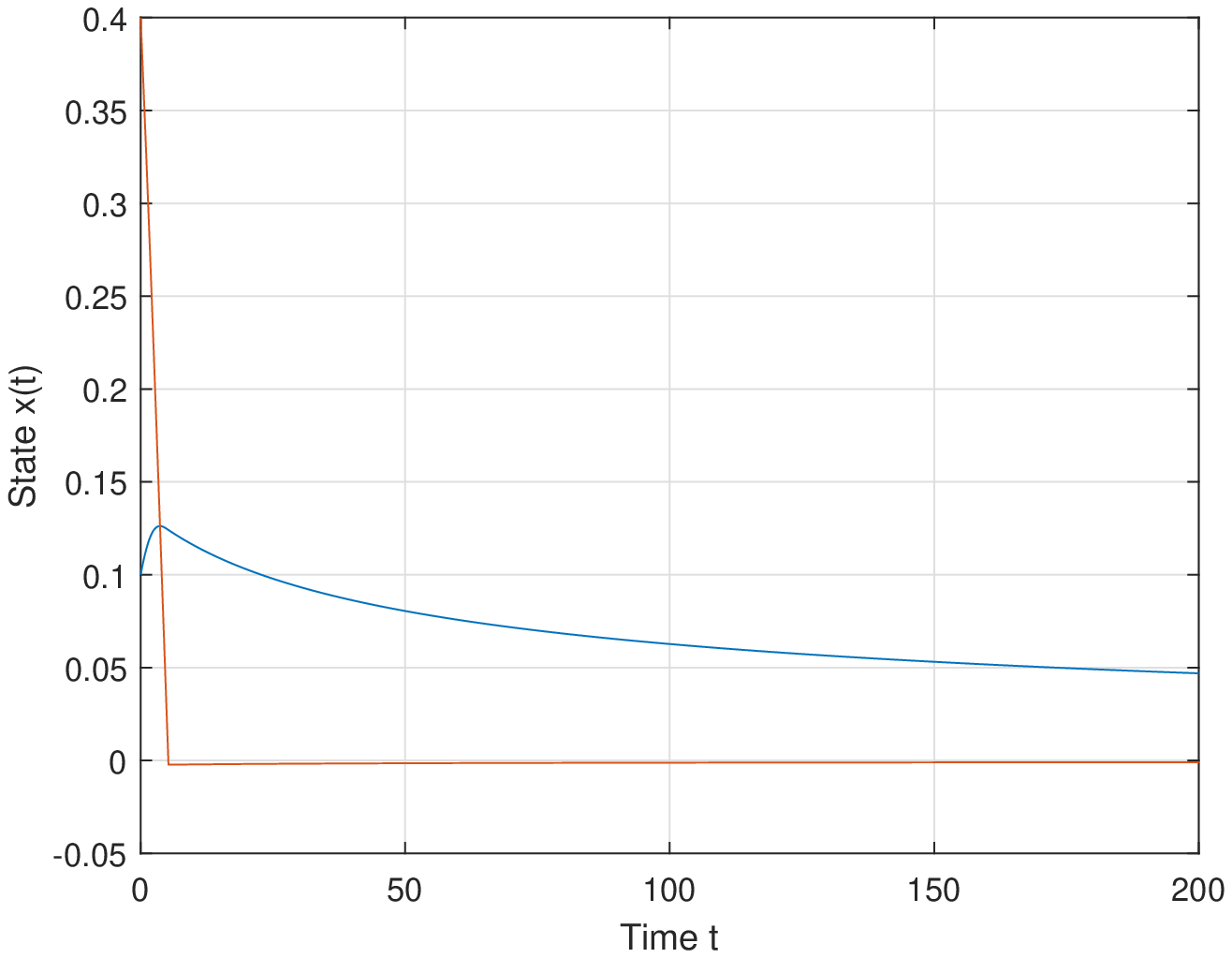}
\end{subfigure}
\begin{subfigure}[t]{0.75\columnwidth}
\center
\includegraphics[width=\columnwidth]{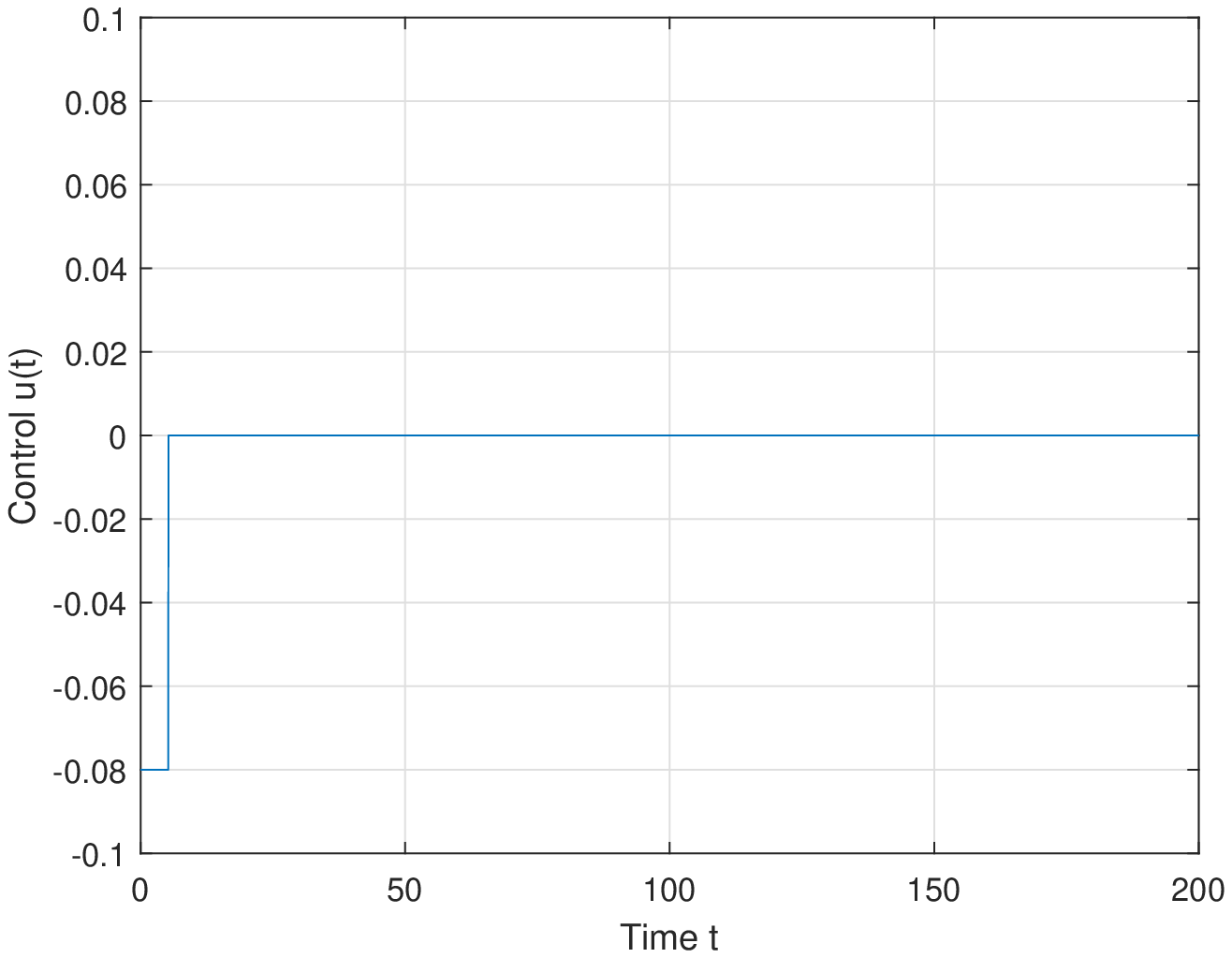}
\end{subfigure}
\begin{subfigure}[t]{0.75\columnwidth}
\center
\includegraphics[width=\columnwidth]{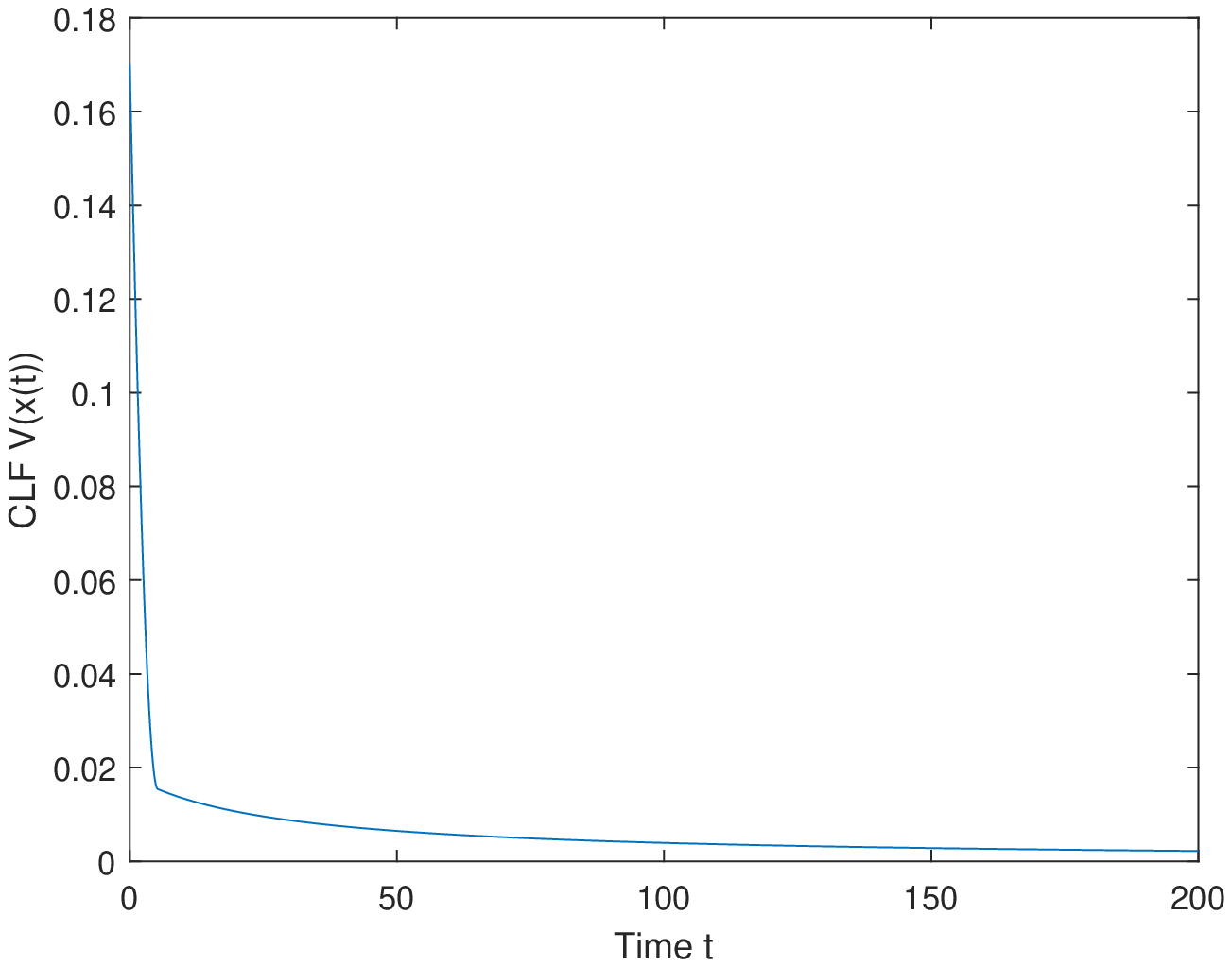}
\end{subfigure}
\caption{Event-triggered stabilization of system~\eqref{eq.syst-anta}}\label{Fig.ex2}
\end{figure}

\section{Conclusion}

In this paper, we address the following fundamental question: let a nonlinear system admit a control Lyapunov function (CLF), corresponding to a continuous-time stabilizing controller. Does this imply the existence of an event-triggered controller, also providing exponential convergence? Under certain natural assumptions, we give an affirmative answer and show that such a controller in fact also provides the positive dwell time between consecutive events, and the convergence rate of the closed-loop system can be arbitrarily close to the continuous time system's rate. The results remain valid for non-exponentially stabilizing CLFs, that provide polynomial convergence rate. Two examples are considered,
illustrating application of the proposed method to nonlinear systems.

Although the existence of CLFs often can be derived from the inverse Lyapunov theorems, to find a CLF satisfying Assumptions~\ref{ass.F}-\ref{ass.non-degen}
is a non-trivial problem; computational approaches to cope with it are subject of ongoing research. 
%Another important problem to be addressed in the future works is the relation between our controller and the one based on Sontag's formula~\cite{Marchand:2013}. The simulation in Section~4.2 shows %that the behaviors of the closed-loop system
%under these controllers are very similar, although they are based on different principles.
Finally, it should be noted that the CLF method is not the only approach to event-triggered control of nonlinear system, e.g. in~\cite{LiuLiuDou:2014} an \emph{impulsive} event-triggered controller,
exponentially stabilizing a nonlinear system, has been proposed. Unlike our controller, this controller leads to \emph{discontinuous} trajectories $x(t)$ (at each sampling instant, the continuous dynamics is stopped and restarted in another point) and imposes some other limitations, e.g. the system has to be fully actuated with globally Lipschitz right-hand side.

\bibliographystyle{ACM-Reference-Format}
\bibliography{event,nonlinear}

\appendix
\section{Proof of Lemma~3.12}\label{sec.proof}

We start with several technical propositions, establishing some useful properties of the solution to a general Cauchy problem~\eqref{eq.cauchy}. Throughout this subsection, we assume that Assumptions~\ref{ass.F}-\ref{ass.non-degen} hold.
Henceforth we always assume that in~\eqref{eq.cauchy} $u_*=\u(x_*)$ and $\xi_0\in B(x_*)$.
Denoting the (unique) solution to~\eqref{eq.cauchy} by $\xi(t|\xi_0,u_*)$, let $t_*(\xi_0,u_*)>0$ stand for the first instant $t$ when $W(\xi(t),u_*)=-\sigma\gamma V(\xi(t))$ and $\Delta_*(\xi_0,u_*)=[0,t_*(\xi_0,u_*)]$. If such an instant does not exist, we put $t_*(\xi_0,u_*)=\infty$ and $\Delta_*=[0,\infty)$.
The solution $\xi(t)=\xi(t|\xi_0,u_*)$ is well defined on $\Delta_*(\xi_0,u_*)$ thanks to Proposition~\ref{prop.tech}, and $\xi(t)\in B(\xi_0)$ since $\dot V(\xi(t))=W(\xi(t),u_0)\le 0$ on $\Delta_*(\xi_0,u_*)$.

\begin{proposition}
For any $\xi_0\in B(x_*)$, $u_*=\u(x_*)$ and $t\in\Delta_*(\xi_0,u_*)$, the vector $\xi(t)=\xi(t|\xi_0,u_*)$ satisfies the inequalities:
\be\label{eq.xi-prop}
\begin{gathered}
|\xi(t)-\xi_0|\le c(t,x_*)|F(\xi_0,u_*)|,\\ |F(\xi(t),u_*)|\le (1+c(t,x_*))|F(\xi_0,u_*)|,\\
c(t,x_*)\dfb \left(\frac{e^{(2\varrho(x_*)+1)t}-1}{2\varrho(x_*)+1}\right)^{1/2}.
\end{gathered}
\ee
Here $\varrho(x_*)$ is the Lipschitz constant from~\eqref{eq.kappa}.
\end{proposition}
\begin{proof}
Let $\alpha(t)\dfb |\xi(t)-\xi_0|^2/2$. By noticing that $\dot\alpha(t)=(\xi(t)-x_*)^{\top}F(\xi(t),u_*)$, one arrives at the inequality
\[
\begin{aligned}
\dot\alpha&(t)=(\xi(t)-\xi_0)^{\top}[F(\xi(t),u_*)-F(\xi_0,u_*)]+\\
&+(\xi(t)-\xi_0)^{\top}F(\xi_0,u_*)\le 2\varrho(x_*)\alpha(t)+\alpha(t)+\frac{|F(\xi_0,u_*)|^2}{2}
\end{aligned}
\]
(by assumption, that $\xi_0\in B(x_*)$, and hence $\xi(t)\in B(x_*)$ for any $t\in\Delta_*(\xi_0,u_*)$).
The usual comparison lemma implies that $\alpha(t)\le\beta(t)$, where $\beta(t)$ is the solution to the Cauchy problem
\[
\dot\beta(t)=[2\varrho(x_*)+1]\beta(t)+\frac{|F(\xi_0,u_*)|^2}{2},\quad \beta(0)=\alpha(0)=0.
\]
Obviously, $\beta(t)=c(t,x_*)^2|F(\xi_0,u_*)|^2/2$, which entails the the first inequality in~\eqref{eq.xi-prop}.
The second inequality is immediate from~\eqref{eq.kappa} since $|F(\xi(t),u_*)|\le |F(\xi_0,u_*)|+\varrho(x_*)|\xi(t)-\xi_0|$.
\end{proof}

To simplify the estimates for the minimal dwell time, we will use the following simple inequality for the function $c(t,x_*)$.
\begin{proposition}
If $0\le t\le (1+2\varrho(x_*))^{-1}$, then
\be\label{eq.c-est}
c(t,x_*)\le \sqrt{te}\le \sqrt{e}.
\ee
\end{proposition}
\begin{proof}
Denoting for brevity $\varrho=\varrho(x_*)$, the statement follows from the mean value theorem, applied to the function $e^{(2\varrho+1)t}$. Since $e^{(2\varrho+1)t}-1=t(2\varrho+1)e^{(2\varrho+1)t_0}$,
$t_0\in (0,t)$,
\ben
c(t,x_*)^2=\frac{e^{(2\varrho+1)t}-1}{2\varrho+1}=te^{(2\varrho+1)t_0}\le te^{(2\varrho+1)t}\le te\le e,
\een
which implies the inequalities~\eqref{eq.c-est}.
\end{proof}
\begin{corollary}
Let $\xi_0\in B(x_*)$, $u_*=\u(x_*)$ and $\xi(t)=\xi(t|\xi_0,u_*)$, where $t\in\Delta_*(\xi_0,u_*)\cap \left[0,(1+2\varrho(x_*))^{-1}\right]$. Then
\be\label{eq.delta_w1}
\begin{aligned}
|W(\xi(t),u_*)&-W(\xi_0,u_*)|\le \\&\le\sqrt{t}\mu(x_*)\left(|V'(\xi_0)|\,|F(\xi_0,u_*)|+|F(\xi_0,u_*)|^2\right),\\
\mu(x_*)&\dfb e^{1/2}\max\left\{\varrho(x_*),(1+\sqrt{e})\nu(x_*)\right\}.
\end{aligned}
\ee
Here $\nu(x_*)$ is the Lipschitz constant from~\eqref{eq.nu}.
\end{corollary}
\begin{proof}
Recalling that $\xi=\xi(t)\in B(\xi_*)$, one has
\[
\begin{aligned}
|W&(\xi(t),u_*)-W(\xi_0,u_*)|\leq |\left(V'(\xi(t))-V'(\xi_0)\right)F(\xi(t),u_*)|+\\
+&|V'(\xi_0)\left(F(\xi(t),u_*)-F(\xi_0,u_*)\right)|\overset{\eqref{eq.kappa},\eqref{eq.nu}}{\leq}\\
\leq & \nu(x_*)|\xi(t)-\xi_0| |F(\xi(t),u_*)|+\varrho(x_*)|V'(\xi_0)| |\xi(t)-\xi_0|\overset{\eqref{eq.xi-prop}}{\leq}\\
\leq & \nu(x_*)c(t,x_*)(1+c(t,x_*))|F(\xi_0,u_*)|^2+\\
+ & \varrho(x_*)c(t,x_*)|V'(\xi_0)|\,|F(\xi_0,u_*)|\overset{\eqref{eq.c-est}}{\leq}\\
\leq &\sqrt{te}\left(\varrho(x_*)|V'(\xi_0)|\,|F(\xi_0,u_*)|+\nu(x_*)(1+\sqrt{e})|F(\xi_0,u_*)|^2\right).
\end{aligned}
\]
The inequality~\eqref{eq.delta_w1} now follows from the definition of $\mu(x_*)$.
\end{proof}

With some abuse of notation, let $t_*(x_*)\dfb t_*(x_*,u_*)$ and $\Delta(x_*)\dfb \Delta(x_*,u_*)$.
Substituting $\xi_0=x_*$ into the inequality~\eqref{eq.delta_w1}, one obtains the following proposition.
\begin{proposition}\label{prop.w}
For an arbitrary $\sigma\in (0,1)$ and $x_*\ne 0$ let
\be\label{eq.tau}
\tilde\tau_{\sigma}(x_*)=\min\left\{\frac{(1-\sigma)^2}{\mu(x_*)^2M(x_*)^2},\frac{1}{1+2\varrho(x_*)}\right\}>0.
\ee
Then $t_*(x_*)\ge\tilde\tau_{\sigma}(x_*)$ and for any $t\in [0,\tilde\tau_{\sigma}(x_*))$ the solution $\xi(t)=\xi(t|x_*,u_*)$ (where $u_*=\u(x_*)$) satisfies the following inequalities
\be\label{eq.w-ineq}
%(2-\sigma)W(x_*,u_*)<
W(\xi(t),u_*)<\sigma W(x_*,u_*)<-\sigma\gamma V(\xi(t)).
\ee
\end{proposition}
\begin{proof}
For any $t\in\Delta_*(x_*)\cap [0,(1+2\varrho(x_*))^{-1})$, one has
\be\label{eq.aux2}
\begin{aligned}
|W(\xi(t),u_*)-W(x_*,u_*)|\overset{\eqref{eq.non-degen1},\eqref{eq.delta_w1}}{\le}\sqrt{t}\mu(x_*)M(x_*)|V'(x_*)\bar F(x_*)|\\
\overset{\eqref{eq.w}}{=}
\sqrt{t}\mu(x_*)M(x_*)|W(x_*,u_*)|.
\end{aligned}
\ee
For any $t<(1-\sigma)^2/(\mu(x_*)M(x_*))^2$ one has $\sqrt{t}\mu(x_*)M(x_*)<1-\sigma$ due to~\eqref{eq.tau}.
Inequality~\eqref{eq.aux2} and definition~\eqref{eq.tau} entail that
\be\label{eq.aux3}
\begin{gathered}
\dot V(\xi(t))=W(\xi(t),u_*)<W(x_*,u_*)+(1-\sigma)|W(x_*,u_*)|=\\
=|W(x_*,u_*)|(-1+1-\sigma)=\sigma W(x_*,u_*)\overset{\eqref{eq.inf-u-exp}}{\le} -\sigma\gamma V(x_*)<0.
\end{gathered}
\ee
whenever $t<\min(t_*,\tau_{\sigma}(x_*))$. By noticing that $V(\xi(t))<V(x_*)$ and thus
$-\sigma\gamma V(\xi(t))>-\sigma\gamma V(x_*)$, one shows that~\eqref{eq.w-ineq} holds for $t<\min(t_*,\tau_{\sigma}(x_*))$. By definition, we either have $W(\xi(t_*),u_*)=-\gamma V(\xi(t_*))$ or $t_*=\infty$; hence $t_*\geq\tilde\tau_{\sigma}(x_*)$, which ends the proof.
\end{proof}

\textbf{Proof of Lemma~\ref{lem.key-lemma}}

Let $\tau_{\sigma}(x)$ stand for the function~\eqref{eq.tau}. The inequality~\eqref{eq.w-ineq-main} follows from Proposition~\ref{prop.w}. Recalling that $\varrho(x),\nu(x),M(x)$ are locally bounded, which is also valid for $\mu(x)$ and $\rho(x)$,~\eqref{eq.tau} implies that $\tau_{\sigma}(x)$ is \emph{uniformly positive} on any compact set.
If the functions $\varrho(x),\nu(x),M(x),\rho(x)$ are \emph{globally bounded}, the same holds for $\mu(x)$ and $\tau_{\sigma}(x)$ is \emph{uniformly positive} over all $x\in\r^n$. $\square$

\end{document}